\documentclass[11pt]{amsart}

\usepackage{amssymb,amsmath,amsthm}
\usepackage{latexsym}
\usepackage{arydshln}
\usepackage{graphicx}
\graphicspath{{figures/}}
\usepackage{tikz-cd}
\usepackage{mathtools,amssymb}
\usepackage{kotex}
\usepackage{fancyhdr,geometry}
\usepackage{caption} 
\newtheorem{theorem}{Theorem}[section]
\newtheorem{thm}{Theorem}[section]

\newtheorem{lemma}[thm]{Lemma}

\newtheorem{proposition}[thm]{Proposition}

\theoremstyle{definition}

\theoremstyle{remark}
\newtheorem{remark}[thm]{Remark}

\numberwithin{equation}{section}

\begin{document}

\title{Effective data reduction algorithm for topological data analysis}

\author{Seonmi Choi}
\author{Jinseok Oh}
\author{Jeong Rye Park}
\author{Seung Yeop Yang}
\author{Hongdae Yun}

\email[Seonmi Choi]{smchoi@knu.ac.kr}
\email[Jinseok Oh]{jinseokoh@knu.ac.kr}
\email[Jeong Rye Park]{parkjr@knu.ac.kr}
\email[Seung Yeop Yang]{seungyeop.yang@knu.ac.kr}
\email[Hongdae Yun]{yyyj1234@knu.ac.kr}

\address{Nonlinear Dynamics and Mathematical Application Center, Kyungpook National University, Daegu 41566, Republic of Korea}
\address{Department of Mathematics, Kyungpook National University, Daegu, 41566, Republic of Korea}
\address{Department of Mathematics, Kyungpook National University, Daegu, 41566, Republic of Korea}
\address{Department of Mathematics, Kyungpook National University, Daegu, 41566, Republic of Korea}
\address{Department of Mathematics, Kyungpook National University, Daegu, 41566, Republic of Korea}

\subjclass{}%

\begin{abstract}
One of the most interesting tools that have recently entered the data science toolbox is topological data analysis (TDA). 
With the explosion of available data sizes and dimensions, identifying and extracting the underlying structure of a given dataset is a fundamental challenge in data science, and TDA provides a methodology for analyzing the shape of a dataset using tools and prospects from algebraic topology. However, the computational complexity makes it quickly infeasible to process large datasets, especially those with high dimensions. Here, we introduce a preprocessing strategy called the Characteristic Lattice Algorithm (CLA), which allows users to reduce the size of a given dataset as desired while maintaining geometric and topological features in order to make the computation of TDA feasible or to shorten its computation time. In addition, we derive a stability theorem and an upper bound of the barcode errors for CLA based on the bottleneck distance.
\end{abstract}

\keywords{topological data analysis, persistent homology, Vietoris-Rips filtration, topology-preserving data reduction}

\subjclass[2020]{Primary: 55N31. Secondary: 62R40, 68T09.}

\maketitle

\section{\bf{Introduction}}

Recent advances in computational science have led to a data explosion in the scientific community exploring complex natural phenomena. In particular, high-dimensional and complex simulations are being implemented in various fields using high-performance computing technology. However, these datasets, especially biological data and astronomical data, are often so high-dimensional that they severely limit our visualization capabilities, as well as have much more noise and missing information.

Topological data analysis (TDA in abbreviation) combines theoretical tools and algorithms in algebraic topology and geometry to enable a mathematically rigorous study of the ``shape" of data \cite{Car09, CarSil10, Ghr08, LSLIVACC13, ZomCar05}. TDA provides dimensionality reduction and robustness for high-dimensional, incomplete, and noisy data. One of the main tools of TDA is persistent homology, which is a method for retrieving the essential topological features of a dataset. In order to calculate the persistent homology of a dataset, one first needs to represent the dataset as a filtered simplicial complex, which is a sequence of simplicial complexes built on the dataset in order to add topological structure to it. One widely used method for doing this is the Vietoris-Rips filtration \cite{Gro87, Vie27}. This filtration, unfortunately, is often too computationally expensive to construct in full (The run-time of computation grows exponentially with the number of simplices. \cite{OPTGH17}). As a result, new approaches to reducing computation costs have recently emerged that employ data reprocessing and sub-sampling \cite{CFLMRW15, MMW18, MSW20, MW19, RVT14, She13, She14, SC04}. Sheehy \cite{She13} introduced an alternative filtration on an $n$-point metric space, called the sparse Vietoris–Rips filtration that has a size $O(n)$ and can be computed in $O(n \log n)$ time. Chazal et al. \cite{CFLMRW15} proposed a method for reducing the computation independent of the filtration by sub-sampling the data randomly repeatedly and creating an average landscape for the point cloud. de Silva and Carlsson \cite{SC04} employ a random or MaxMin approach to select ``landmark" points that are near other members in the point cloud. Moitra et al. \cite{MMW18} use the centroids of $k$-means++ clusters in order to reduce the number of data points.\\

In this paper, we introduce a method for reducing data size while preserving its topological features. The computation time for the persistent homology of a given dataset is generally exponential with respect to the number of points in the data, which requires significant computational costs for large datasets. The proposed method allows for a significant reduction in computation time while having a minimal impact on the output results. Furthermore, the reduction ratio of the input data can be adjusted as desired, making it possible to perform the calculation for the persistent homology on big data. This method can be utilized as a preprocessing step to reduce the size of data before inputting it into persistent homology algorithms such as Ripser \cite{Bau21}. Therefore, it can be combined with other methods that reduce the computational cost of persistent homology. The cost-saving effect and the accuracy of this method will be analyzed using data generated through the NumPy library.\\

This paper is organized as follows. Section \ref{IntroTDA} describes the basic notions and ideas of TDA, particularly on persistent homology. Section \ref{Methods} introduces a new data reduction algorithm suitable for TDA. Sections \ref{Data} to \ref{R_acc} analyze the effects of this algorithm in terms of computational cost reduction and accuracy of data analysis. Finally, Section \ref{Conclusion} presents our conclusions.

\vspace{1cm}
\section{\bf{Background: Topological Data Analysis}}\label{IntroTDA}

We briefly review the basic definitions and notions of persistent homology of simplicial complexes obtained from point clouds. See \cite{EdeHar10, RabBlu20} for more details.\\

{\bf Filtered Vietoris-Rips Complex}\, Let $(X, d_{X})$ be a finite metric space and let $\varepsilon >0$ be fixed.
The {\it Vietoris-Rips complex} $VR_{\varepsilon}(X, d_{X})$ is the abstract simplicial complex such that
  \begin{itemize}
    \item[(1)] vertices are the elements of $X$,
    \item[(2)] $[v_{0}, v_{1}, \dots, v_{k}]$ is a $k$-simplex if $d_{X}(v_{i}, v_{j})\leq 2\varepsilon$ for all $0\leq i, j \leq k$, where $\{v_{0}$, $v_{1}$, $\dots$, $v_{k}\}$ is a subset of $X$.
  \end{itemize}
For $\varepsilon<\varepsilon'$, there exists the induced simplicial map 
$VR_{\varepsilon}(X, d_{X})\rightarrow VR_{\varepsilon'}(X, d_{X})$ as every simplex of $VR_{\varepsilon}(X, d_{X})$ is also a simplex of $VR_{\varepsilon'}(X, d_{X})$.
Since $X$ is finite, one can choose finite values $\varepsilon_{1}< \varepsilon_{2} < \cdots < \varepsilon_{m}$ so that $VR_{\varepsilon_{i}}(X, d_{X}) \subsetneq VR_{\varepsilon_{i+1}}(X, d_{X})$ for each $i.$
Then the natural inclusions 
$$VR_{\varepsilon_{1}}(X, d_{X}) \hookrightarrow VR_{\varepsilon_{2}}(X, d_{X}) \hookrightarrow \cdots \hookrightarrow VR_{\varepsilon_{m}}(X, d_{X})$$ 
imply the induced homomorphisms
$$PH_{k}(VR_{\varepsilon_{1}}(X, d_{X}); \mathbb{F}) \rightarrow PH_{k}(VR_{\varepsilon_{2}}(X, d_{X}); \mathbb{F}) \rightarrow \cdots \rightarrow PH_{k}(VR_{\varepsilon_{m}}(X, d_{X}); \mathbb{F}) $$ 
by the $k$th homology group functor $PH_{k}$, where $\mathbb{F}$ is a field.
The sequence of simplicial complexes $VR_{\bullet}(X, d_{X})$ is called a {\it filtered Vietoris-Rips complex} on $X.$ \\

{\bf Barcode and Persistent Diagram}\,
For a filtered Vietoris-Rips complex $VR_{\bullet}(X, d_{X})$, an element $\gamma \in PH_{k}(VR_{\varepsilon_{i}}(X, d_{X}); \mathbb{F})$ is 
  \begin{itemize}
    \item[(1)] {\it born at $\varepsilon_{i}$} if $\gamma \notin \textrm{Im}(\theta_{j,i})$ for any $\varepsilon_{j}<\varepsilon_{i}$,
     \item[(2)] {\it dies at $\varepsilon_{l} >\varepsilon_{i}$} if $\theta_{i,l}(\gamma)=0$ in $PH_{k}(VR_{\varepsilon_{l}}(X, d_{X}); \mathbb{F})$ or for some $\varepsilon_{j}<\varepsilon_{l}$, there exists an element $\alpha \in PH_{k}(VR_{\varepsilon_{j}}(X, d_{X}); \mathbb{F})$ such that $\theta_{i,l}(\gamma)=\theta_{j,l}(\alpha)$,
  \end{itemize}
where $\theta_{a,b}$ denotes the induced homomorphism $PH_{k}(VR_{\varepsilon_{a}}(X, d_{X}); \mathbb{F}) \rightarrow PH_{k}(VR_{\varepsilon_{b}}(X, d_{X}); \mathbb{F})$ for any $\varepsilon_{a}<\varepsilon_{b}$.
The interval $[\varepsilon_{i},\varepsilon_{l})$ represents the lifespan of $\gamma$.
The $k$th {\it barcode} of $X$, denoted by $B_{k}X$, can be defined as the multiset of non-empty intervals, called {\it bars}, of the form either $[x, y)$ or $[x, \infty)$ and the bars representing the lifespans of particular elements in $PH_{k}(VR_{\bullet}(X, d_{X}); \mathbb{F}).$ It is often efficient to consider a barcode as a set of points in $\mathbb{R}^{2},$ called a {\it persistent diagram} (PD in abbreviation), by replacing each bar $[x, y)$ of the barcode with the point $(x, y) \in \mathbb{R}^{2}.$ Figure \ref{TDA}, for example, illustrates the process of obtaining a Vietoris-Rips complex, a barcode, and a persistent diagram of the given point cloud $X$ consisting of nine points in $\mathbb{R}^2.$

\begin{figure}[h!]
  \centering
  \includegraphics[width = 15cm]{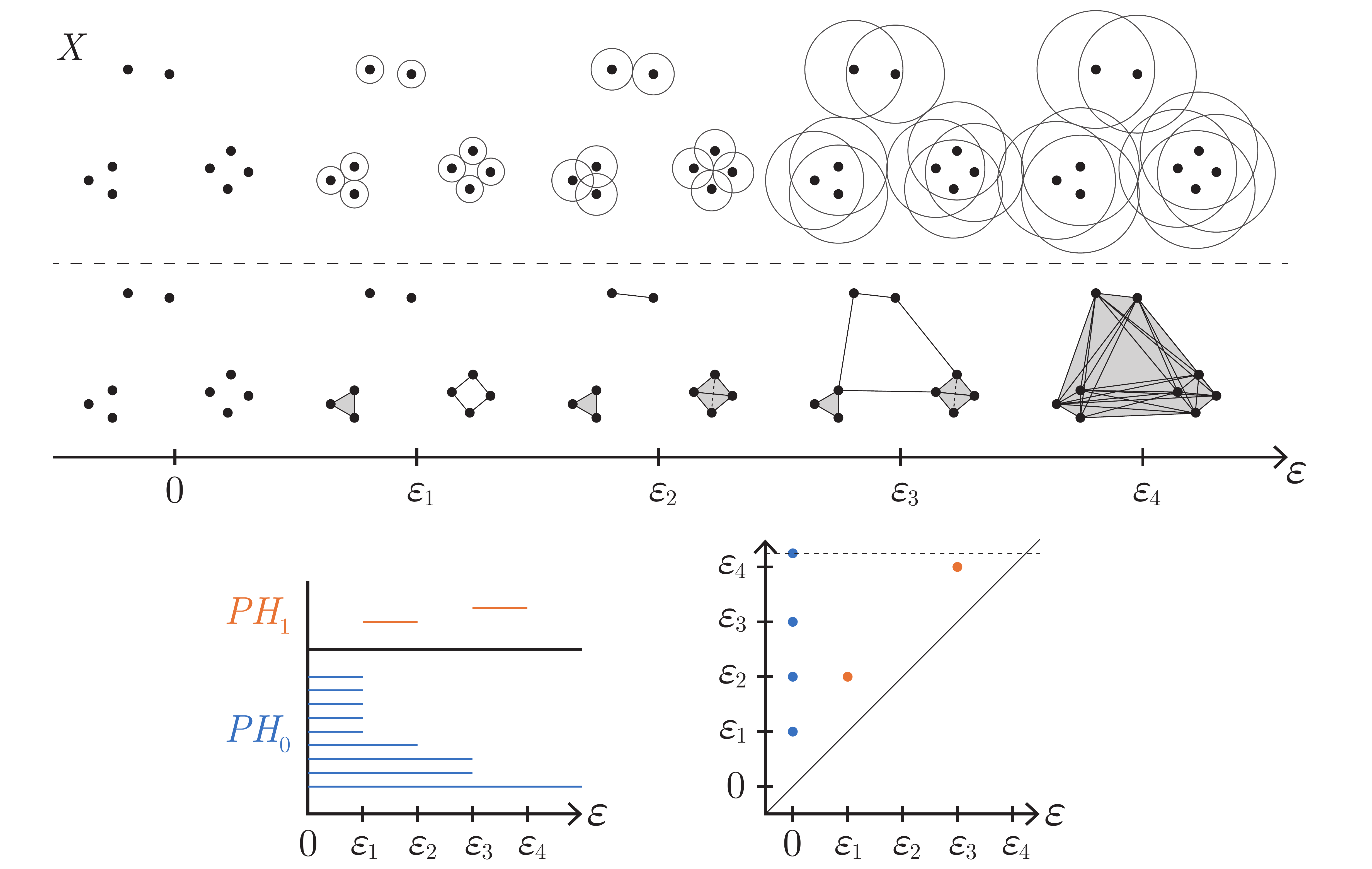}
  \caption{{\bf A barcode and a persistent diagram}}\label{TDA}
\end{figure}


For the stability of persistent homology under perturbation, we use the distance between barcodes to be precise about measuring changes in the output of TDA.\\

{\bf The Gromov-Hausdorff Distance}\, 
Let $(X,d_{X})$ and $(Y,d_{Y})$ be two compact metric spaces.
The {\it Gromov-Hausdorff distance} between $X$ and $Y$ is defined by
$$d_{GH}((X,d_{X}), (Y,d_{Y}))=\underset{\theta_{1}, \theta_{2}}{\textrm{inf}} d_{H}(X, Y),$$
where $\theta_{1}:X\rightarrow Z$ and $\theta_{2}:Y\rightarrow Z$ are isometric embeddings of $(X,d_{X})$ and $(Y,d_{Y})$ in any metric space $(Z, d_{Z})$ respectively and $d_{H}$ is the Hausdorff distance.\\

{\bf The Bottleneck Distance}\, 
For two intervals $[x_{1},y_{1})$ and $[x_{2},y_{2}),$ we define
$$d_{\infty}([x_{1},y_{1}),[x_{2},y_{2}))=\textrm{max}\{|x_{1}-x_{2}|,|y_{1}-y_{2}|\}.$$
For the empty set $\emptyset$, we extend $d_{\infty}$ by defining
$$d_{\infty}([x,y),\emptyset)=\frac{|y-x|}{2}.$$
Let $B_{1}$ and $B_{2}$ be barcodes. We add $\emptyset$ to $B_{1}$ and $B_{2},$ respectively and assume that $|B_{1}|<|B_{2}|$ without loss of generality.
A {\it matching} is a bijection $\phi:A_{1} \rightarrow A_{2},$ where for each $i=1,2,$ $A_{i}$ is a multi-subset of $B_{i}$ and the elements of $B_{i} \setminus A_{i}$ are considered to be matched with $\emptyset.$
The {\it bottleneck distance} between $B_{1}$ and $B_{2}$ is defined to be
$$d_{B}(B_{1}, B_{2})=\underset{\phi}{\textrm{inf}} \underset{I\in B_{1}}{\textrm{sup}} d_{\infty}(I, \phi(I)),$$
where $\phi$ varies over all matchings between $B_{1}$ and $B_{2}$ and the supremum is taken over all bars $I$ in $B_{1}.$ 

The bottleneck distance quantifies the maximum difference in the best matching between the two barcodes. In other words, two barcodes are considered similar in the bottleneck distance when after ignoring ``short" bars, the endpoints of matching ``long" bars are close. The following is the stability theorem for barcodes with respect to the bottleneck distance.

\begin{proposition}\cite{CC-SGMO09, C-SEH07}\label{PropBD}
Let $(X,d_{X})$ and $(Y,d_{Y})$ be finite metric spaces.
Then for all $k\geq 0,$
$$d_{B}(B_{k}X, B_{k}Y) \leq d_{GH}((X,d_{X}), (Y,d_{Y})).$$
\end{proposition}

\vspace{1cm}
\section{\bf{Results}}\label{DM}

In this section, we introduce a novel method for reducing the size of a dataset as desired while maintaining geometric and topological features, in order to enable feasible computation of TDA or to reduce its computation time. Moreover, we investigate the extent to which the calculation speed of TDA improves as well as how well the accuracy of data analysis utilizing TDA is maintained when applying this method, by using random data.

\subsection{\bf{Topology-Preserving Data Reduction: Characteristic Lattice Algorithm}}\label{Methods}

Let $X$ be a point cloud in the $m$-dimensional Euclidean space $\mathbb{R}^{m}$ and let $\delta>0$ be fixed.
The procedure of our method is outlined as follows:
\begin{enumerate}
    \item Divide $\mathbb{R}^{m}$ into $m$-dimensional hypercubes (also simply called $m$-cubes) with side length $\delta$. To be more precise, we consider the partition $\mathbb{R}^{m} / {\sim}$ of $\mathbb{R}^{m}$ induced by the equivalence relation $\sim$ on $\mathbb{R}^{m}$ defined by 
    $$(x_{1}, x_{2}, \ldots, x_{m}) \sim (y_{1}, y_{2}, \ldots, y_{m}) \iff \textrm{ for each } i,\ x_{i}, y_{i} \in [k_{i}\delta, (k_{i}+1)\delta) \textrm{ for some } k_{i} \in \mathbb{Z},$$
    and then the corresponding equivalence classes are $m$-cubes.
    \item For each $m$-cube $C$, if $C \cap X \ne \emptyset,$ then we select one sample point $\mathbf{x}_{C}^{*}$ from $C$ (the sample point does not need to be an element of $C \cap X$). Otherwise, we do not select any sample points. See Figure \ref{ExaCLA1} for example.
    \item The set of all the sample points obtained in (2) forms a new point cloud, which is denoted by $X_{\delta}^{*}.$ See Figure \ref{ExaCLA2} for example.
\end{enumerate}

\begin{figure}[h!]
  \centering
  \includegraphics[width = 15cm]{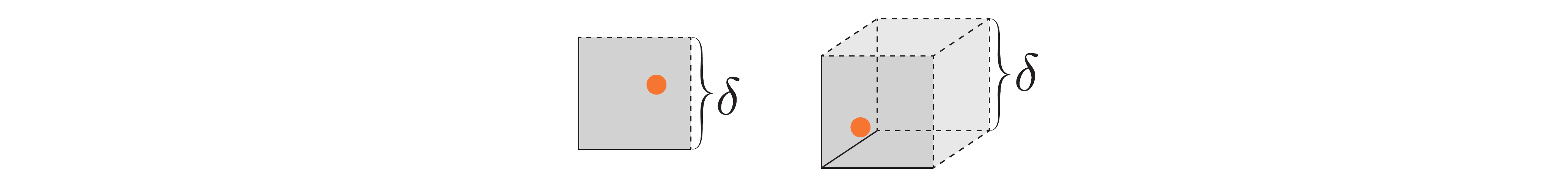}
  \caption{{\bf Sample points from a $2$-cube and a $3$-cube}}\label{ExaCLA1}
\end{figure}

\begin{figure}[h!]
  \centering
  \includegraphics[width = 15cm]{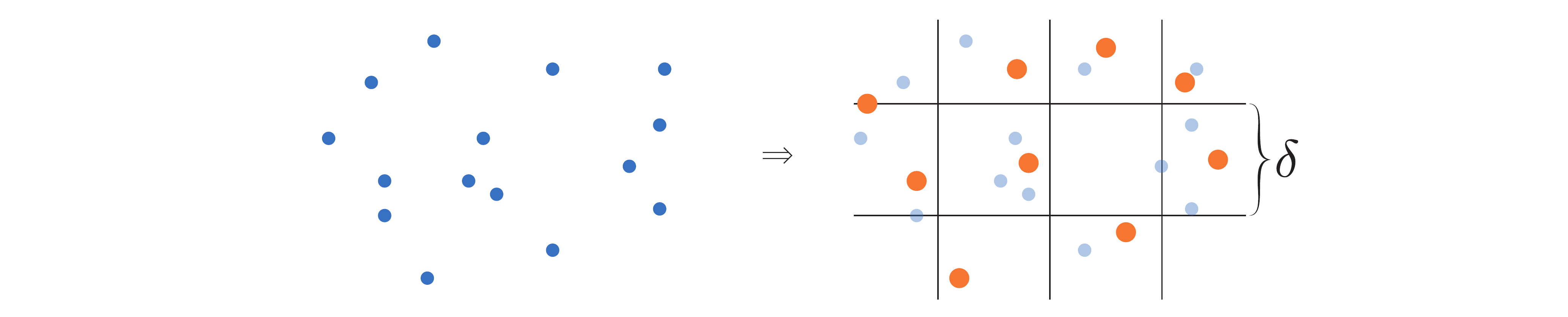}
  \caption{{\bf $X_{\delta}^{*}$(the set of orange points) derived from $X$ using CLA}}\label{ExaCLA2}
\end{figure}

It is obvious that as the value of $\delta$ increases, the size of the new point cloud $X_{\delta}^{*}$ decreases. We call this data reduction technique the {\it Characteristic Lattice Algorithm} (CLA in abbreviation) {\it with $\delta$}, which allows us to shorten the run-time for calculating a persistent diagram while preserving the topological features of the original point cloud $X.$ 

\begin{remark}
When implementing CLA, one can choose all sample points as the centers of the $m$-cubes, and the outcome derived from this selection is denoted by $X_{\delta}^{c}.$ That is,
$$X_{\delta}^{c}=\{(\overline{x}_{1}, \overline{x}_{2}, \ldots, \overline{x}_{m}) \in \mathbb{R}^{m} ~|~ (x_{1}, x_{2}, \ldots, x_{m})\in X\},$$
where $\overline{x}_{i} = \delta \left[ \frac{x_{i}}{\delta} \right] + \frac{\delta}{2}$ for each $i=1, 2, \ldots, m.$
\end{remark}
 
The following theorem supports that CLA preserves topological features of a given data well and that the choice of sample points does not significantly affect the results of CLA.

\begin{theorem}\label{Thm}
Let $X$ be a finite point cloud in $\mathbb{R}^{m}$ and let $\delta > 0$ be fixed.
Then for all $k\geq 0,$
\begin{itemize}
  \item[(1)] $d_{B}(B_{k}X, B_{k}X_{\delta}^{*}) \leq \sqrt{m}\delta;$
  \item[(2)] $d_{B}(B_{k}X_{\delta}^{*}, B_{k}X_{\delta}^{c}) \leq \frac{\sqrt{m}\delta}{2}.$
\end{itemize}
\end{theorem}

\begin{proof}
Note that the Hausdorff distance between $A$ and $B$ is defined as $$d_{H}(A, B)={\textrm{inf}}\{\varepsilon >0  ~|~  B\subseteq A_{\varepsilon}, A\subseteq B_{\varepsilon}\},$$
where $A_{\varepsilon} = \underset{a \in A}{\bigcup} \{x \in X  ~|~ d_{X}(x,a) \leq \varepsilon \}   $ and $B_{\varepsilon} = \underset{b \in B}{\bigcup} \{x \in X  ~|~ d_{X}(x,b) \leq \varepsilon \}$. 
When $\varepsilon = \sqrt{m}\delta$, it is obvious that $X_{\delta}^{*}\subseteq X_{\varepsilon}$ and $X\subseteq (X_{\delta}^{*})_{\varepsilon}$.
Similarly, if $\varepsilon = \frac{\sqrt{m}\delta}{2}$, then $X_{\delta}^{c}\subseteq (X_{\delta}^{*})_{\varepsilon}$ and $X_{\delta}^{*}\subseteq (X_{\delta}^{c})_{\varepsilon}$.
Therefore, $$d_{H}(X, X_{\delta}^{*})\leq \sqrt{m}\delta ~\textrm{ and }~ d_{H}(X_{\delta}^{*},X_{\delta}^{c})\leq \frac{\sqrt{m}\delta}{2}.$$
By the definition of Gromov-Hausdorff distance, we have 
$$ d_{GH}(X, X_{\delta}^{*}) \leq \sqrt{m}\delta ~\textrm{ and }~ d_{GH}(X_{\delta}^{*}, X_{\delta}^{c}) \leq \frac{\sqrt{m}\delta}{2}.$$
Then the inequalities are obtained immediately by Proposition \ref{PropBD}, as desired.
\end{proof}

Let $p_{k}:\mathbb{R}^{m}\rightarrow \mathbb{R}$ denote the $k$th projection defined by $p_{k}(x_{1}, \ldots, x_{k}, \ldots, x_{m})=x_{k}$, for each $k=1, \cdots, m$.
Let $M_{k}$ and $m_{k}$ denote the maximum and the minimum of $p_{k}(X)$, respectively.
Then the following is obtained by using the Pigeonhole principle. 

\begin{lemma}\label{DeltaRange}
If $\delta$ satisfies
$\prod\limits_{k=1}^{m}\left( \left[ \frac{M_{k}}{\delta} \right]+\left[ \frac{-m_{k}}{\delta} \right] +2 \right) < |X|,$
then $|X_{\delta}^{*}|<|X|.$ 
\end{lemma}
Note that CLA is meaningful to apply only when the number of points in $X$ is reduced. 
Therefore, it is important to consider the extent to which the data size is reduced.  
From now on, we will consider the {\it reduction rate} of the number of $X$, denoted by $r$, and use the corresponding $\delta$ for each $r.$
By changing the reduction rate $r$, we can observe the extent to which the computation time is reduced when CLA is applied.

\subsection{\bf{Data}}\label{Data}



We construct the dataset using NumPy 1.24 in Python 3.9, which consists of $n$ pairs of two point clouds in $\mathbb{R}^{m}$.
Both point clouds have the same number of points, let us say $p$.
For a fixed integer $s$ as the random seed, two types of point clouds, namely \textit{sphere type} and \textit{random type}, are formed as follows:  
\begin{enumerate}
    \item Sphere type: the data having $p$ points in 
an $(m-1)$-dimensional sphere shape with a noise $5\times 0.1^{m-1}$ constructed by using the $(m-1)$-dimensional spherical coordinate system;
    \item Random type: the data consisting of uniformly random $p$ points.
\end{enumerate}
One can normalize them in the $m$-dimensional box $[0,100]^{m}$ in order to have the same scale.
Thus, for each seed $s$, there is a pair of two point clouds.
As $s$ ranges from $0$ to $n-1$, we can get $n$ pairs of two point clouds. 
That is, the size of the dataset is $2n$. 
For example, Figure \ref{ExaData} shows a 2-dimensional (2D) dataset and a 3-dimensional (3D) dataset when $n=1$.
\begin{figure}[h!]
  \centering
  \includegraphics[width = 8cm]{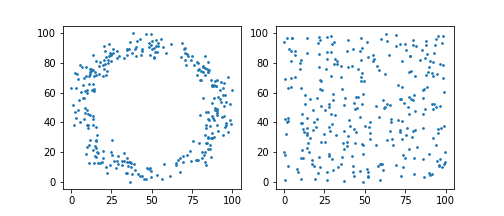}
  \includegraphics[width = 8cm]{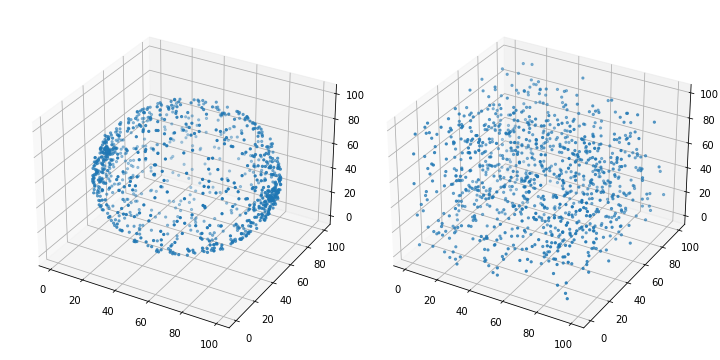}  
  \caption{{\bf A 2D dataset ($p=300$) and a 3D dataset ($p=1,000$)}} \label{ExaData}
\end{figure}
In Figure \ref{ExaDataCLA}, the datasets whose points are colored by orange can be constructed from the datasets in Figure \ref{ExaData} by applying CLA with $\delta=10$ to the 2D dataset and $\delta=20$ to the 3D dataset, respectively, using the centers of the $m$-cubes as the sample points.
\begin{figure}[h!]
  \centering
  \includegraphics[width = 8cm]{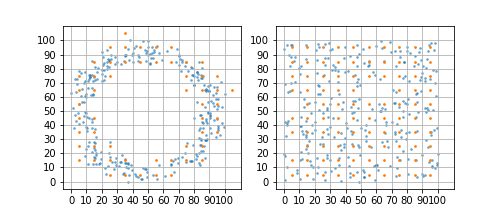}
  \includegraphics[width = 8cm]{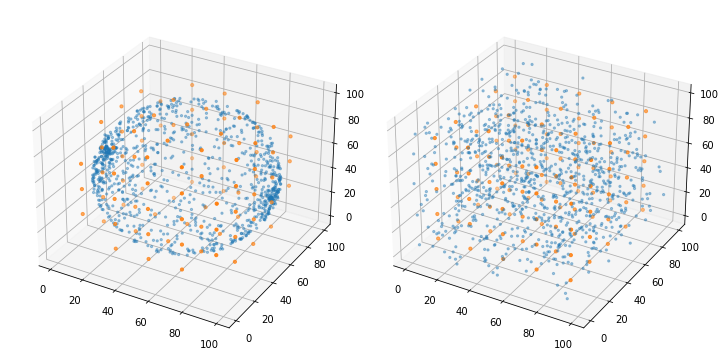}
  \caption{{\bf Data reduction of the datasets in Figure \ref{ExaData} by applying CLA} 
  The point clouds consisting of orange points can be obtained from the point clouds consisting of blue points by applying CLA with $\delta=10$ to the 2D dataset and $\delta=20$ to the 3D dataset.
}\label{ExaDataCLA}
\end{figure}\\

To evaluate the efficiency and effectiveness of CLA, it is necessary to analyze the computation time and classification accuracy.
The workflow on the left in Figure \ref{workflow} represents a pipeline that measures the time required to find a persistent diagram when applying and not applying CLA.
One can evaluate the effectiveness of CLA using a support vector machine (SVM) \cite{CV95} on the dataset.
The workflow on the right in Figure \ref{workflow} shows a pipeline for comparing the accuracy of data analysis with and without CLA.
The hardware environment was 12th Gen Intel(R) Core(TM) i7-12700K 3.61 GHz, 128 GB RAM.

\begin{figure}[h!]
  \centering
  \includegraphics[width = 10cm]{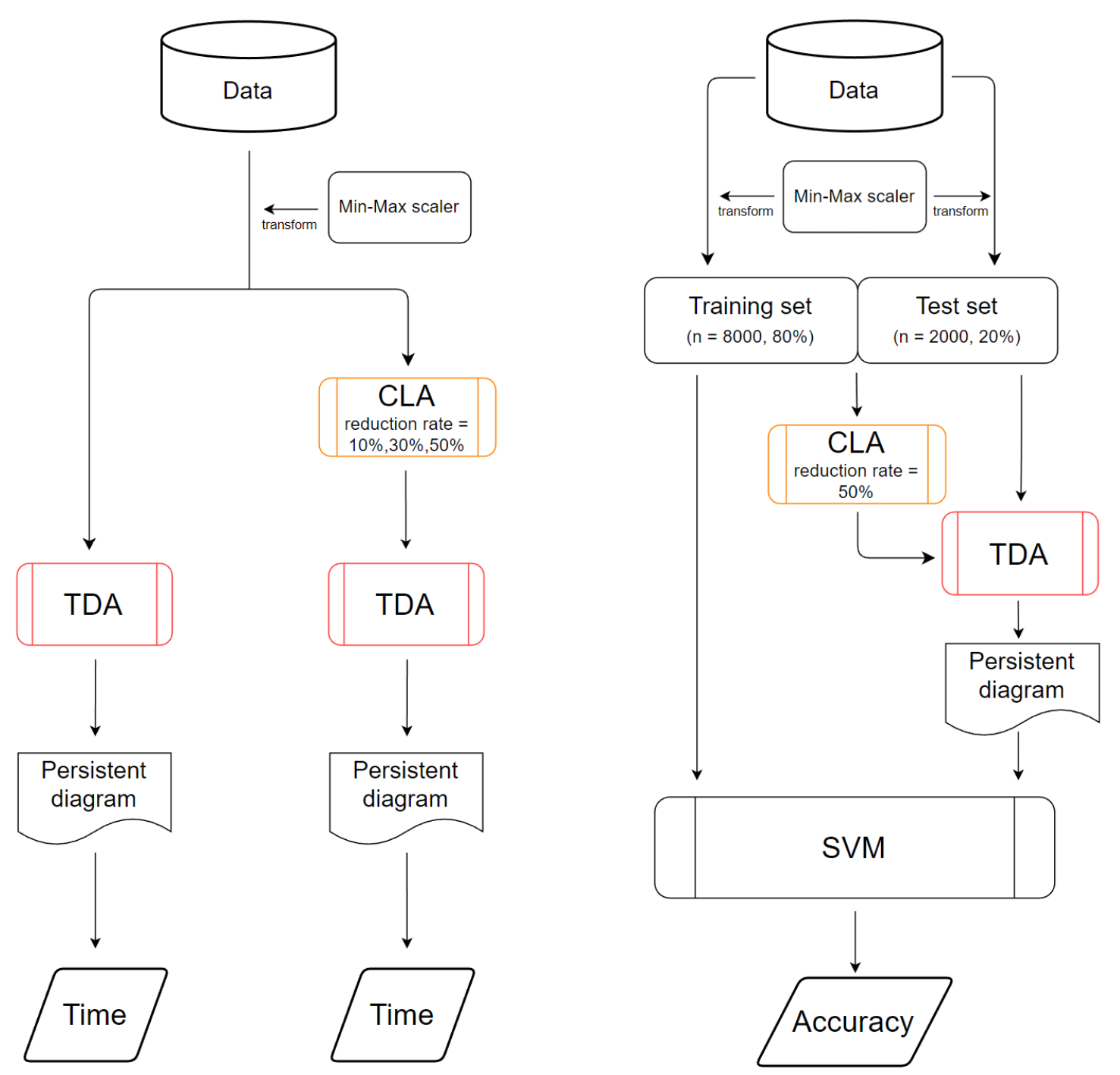}
\caption{{\bf Workflows for measuring computation time and analysis accuracy}
The left workflow compares the computation time required to obtain the PD when applying CLA with the reduction rates of $r=10\%$, $30\%$, or $50\%$, as well as without applying CLA.
The right workflow provides accuracy for classifying two types of point clouds in the dataset when applying CLA with a reduction rate of $r=50\%$ compared to not applying CLA.} \label{workflow}
\end{figure}

\subsection{\bf{Efficiency in Computation Time}}\label{R_time}

To assess the efficiency of calculating a PD with CLA, we will use a random type point cloud as described in Section \ref{Data}.
Figure \ref{(time per pt)} illustrates the run-time for calculating the PD of the 2D and the 3D point clouds. 
Note that $p$ increases by $1,000$ from $1,000$ to $30,000$ in the 2D case and $p$ increases by $100$ from $100$ to $3,000$ in the 3D case.
We compare the computation time it takes to calculate the PD of the original point cloud and the time it takes to calculate the PD of that point cloud after applying CLA with the reduction rates of $r=10\%$, $30\%$, or $50\%$.

\begin{figure}[h!]
\centering
\includegraphics[width = 10cm]{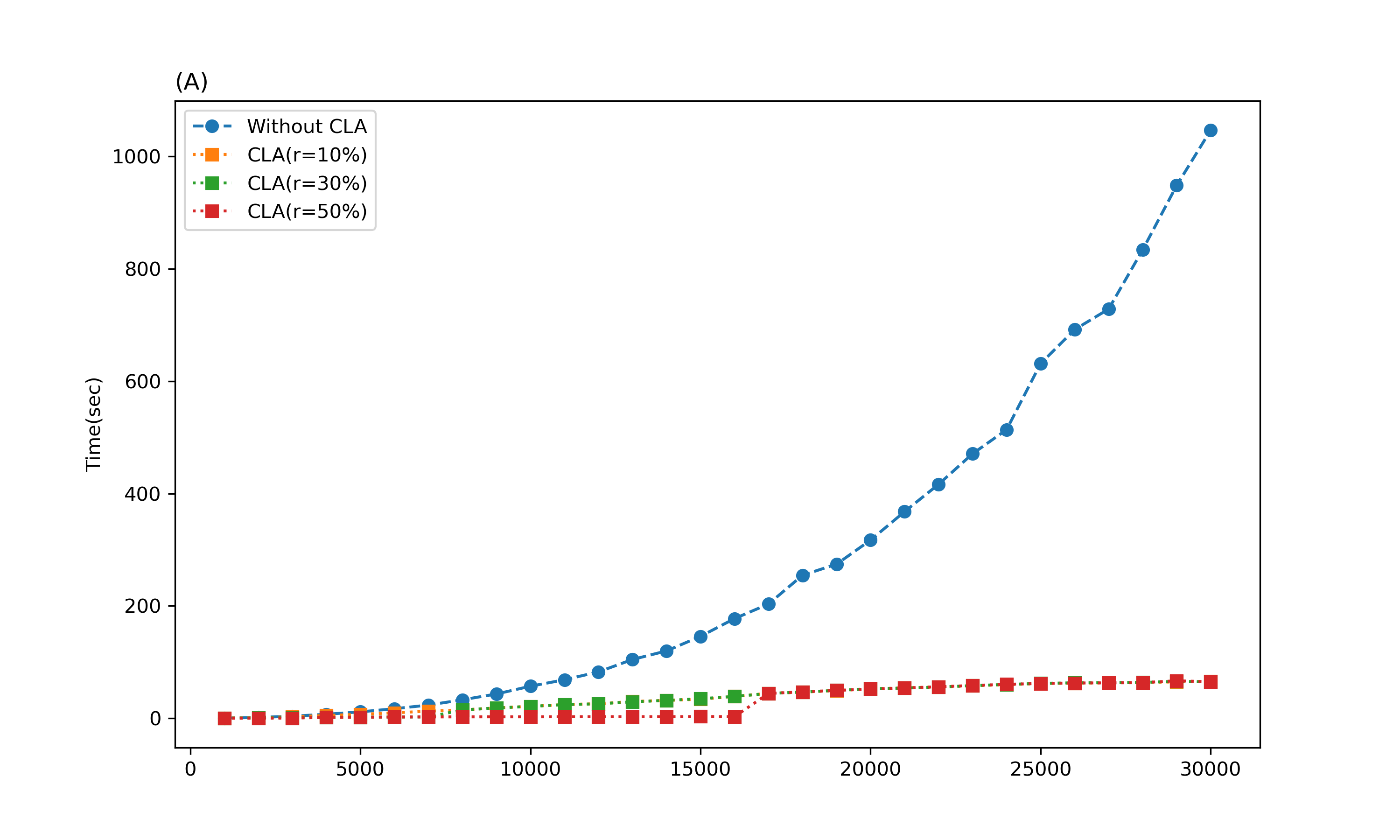}
\includegraphics[width = 10cm]{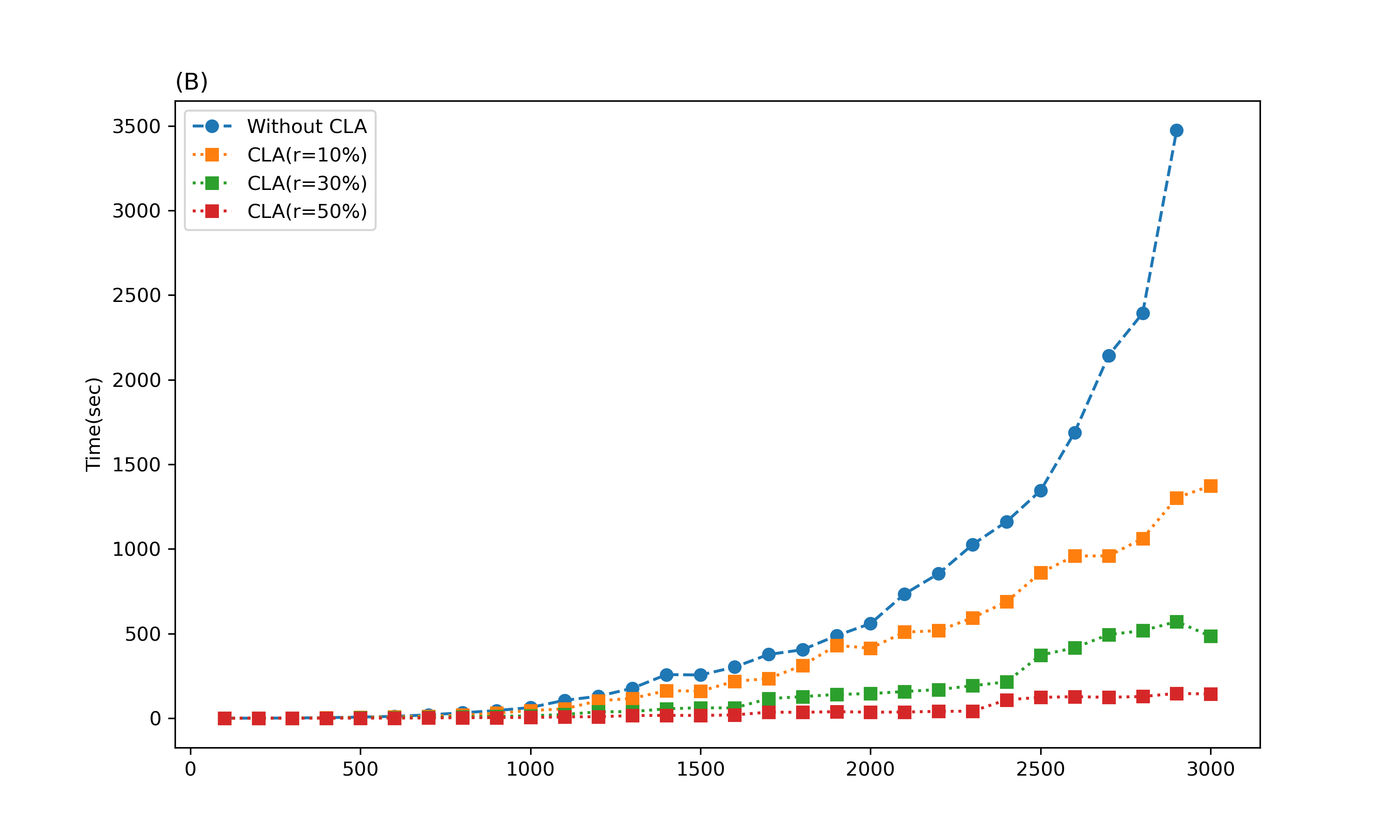}
\caption{{\bf PD Calculation Time Graphs}
The run-time of calculating the PD from the 2D point cloud and the 3D point cloud is presented in the above graphs, where $p$ increases by $1,000$ from $1,000$ to $30,000$ in the 2D point cloud and $p$ increases by $100$ from $100$ to $3,000$ in the 3D point cloud, respectively. (Note that when $p$  reaches $3,000$, due to excessive computation time in our hardware environment, we were unable to calculate the PD without applying CLA.)
The blue dashed line represents the calculation time of the PD, while the other lines indicate the calculation time of the PD after applying CLA with the reduction rates $10\%, 30\%$, and $50\%$.
} \label{(time per pt)}
\end{figure}

As shown in the graph (A) of Figure \ref{(time per pt)}, there is a noticeable difference in computation time between applying and not applying CLA as the size of the 2D data increases. In particular, the information when $p = 30,000$ are shown in Table \ref{Table2Dtime}.
The reduction rates in computation time when applying CLA with the reduction rates $10\%, 30\%$, and $50\%$ are greater than $93\%$ in all cases.
\begin{table}[h!]
    \centering
    \begin{tabular}{c||c|c|c|c}
        $p=30,000$ & Without CLA & CLA($r=10\%$) & CLA($r=30\%$) & CLA ($r=50\%$) \\
         \hline
       Computation Time(sec)  & 1046.5445 & 65.9021 & 64.6348 & 64.9718 \\
       Rate of Decrease($\%$) & $-$ & 93.7 & 93.8 & 93.8 \\ 
    \end{tabular}
    \caption{Reduction rates of calculation time when applying CLA to a 2D point cloud}\label{Table2Dtime}
\end{table}

The graph (B) of Figure \ref{(time per pt)} shows that the utilization of CLA in the 3D point cloud effectively reduces the computation time compared to not using CLA.
Specifically, the information of the calculation time when $p = 2,900$ are shown in Table \ref{Table3Dtime}. 
\begin{table}[h!]
    \centering
    \begin{tabular}{c||c|c|c|c}
        $p=2,900$ & Without CLA & CLA($r=10\%$) & CLA($r=30\%$) & CLA ($r=50\%$) \\
         \hline
       Computation Time(sec)  & 3474.1446 & 1302.1153 & 571.3973 & 146.7123 \\
       Rate of decrease($\%$) & $-$ & 62.5 & 83.5 & 95.7 \\
    \end{tabular}
    \caption{Reduction rates of calculation time when applying CLA to a 3D point cloud}\label{Table3Dtime}
\end{table} 
The reduction rates in computation time when applying CLA with the reduction rate $50\%$ is $95.7\%$.
In particular, when the size of the 3D data is $p = 3,000$, it was not possible to obtain the PD without applying CLA due to our memory constraints and computational limitations. However, when CLA is applied, the PD can be successfully calculated even when $p = 3,000$.

\subsection{\bf{Preservation of Classification Accuracy}}\label{R_acc}
The goal of this section is to observe the changes in the original dataset and its PD when applying CLA, and to compare the data classification accuracy.
To compare the data classification accuracy when CLA is applied and when it is not applied, we analyze the results of the right workflow in Figure \ref{workflow} for the sphere type and random type point clouds presented in Section \ref{Data}. 
The performance of SVM is intricately tied to the careful selection of hyperparameters. 
Our goal is to facilitate a straightforward and unbiased comparison of the SVM performance for two types of point clouds.
Therefore, we deliberately refrained from adjusting any parameters in the SVM employed.\\

First, we analyze how the 2D and 3D datasets, as well as their corresponding PDs, change when applying CLA with the reduction rates $10\%, 30\%$, and $50\%$.
The first rows in Figure \ref{Data_PD_CLA_2D} and Figure \ref{Data_PD_CLA_3D} are the original point clouds and their corresponding PDs without applying CLA.
The second, third, and fourth rows in Figure \ref{Data_PD_CLA_2D} and Figure \ref{Data_PD_CLA_3D} represent the reconstructed datasets obtained by applying CLA with the reduction rates $10\%, 30\%$, and $50\%$ respectively to the original datasets, along with their corresponding PDs.
\begin{figure}[h!]
  \centering
  \includegraphics[width = 8.5cm]{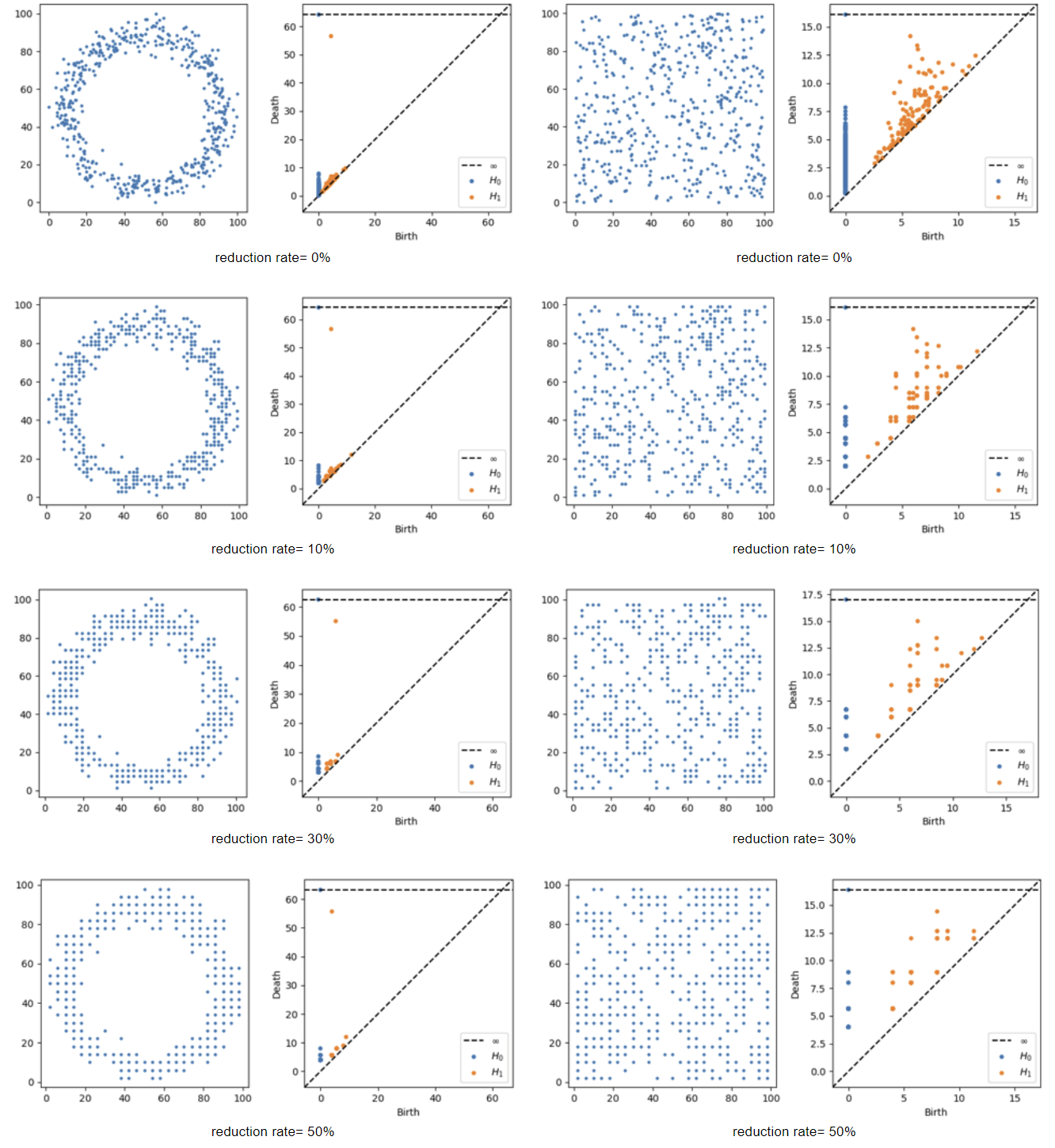}
  \caption{{\bf Evolution of the PDs for a 2D dataset} The PDs of point clouds of a 2D dataset applying CLA with the reduction rates $10\%, 30\%$, and $50\%$.}  \label{Data_PD_CLA_2D}
\end{figure}
When ignoring the points near the diagonal in the PDs, which can be considered as noise, it is evident that the features of the first persistent homology (orange points) of the 2D original datasets and the second persistent homology (green points) of the 3D original datasets are well-preserved in the PDs, despite the increase in the reduction rate of CLA.\\

\begin{figure}[h!]
  \centering
  \includegraphics[width = 8.5cm]{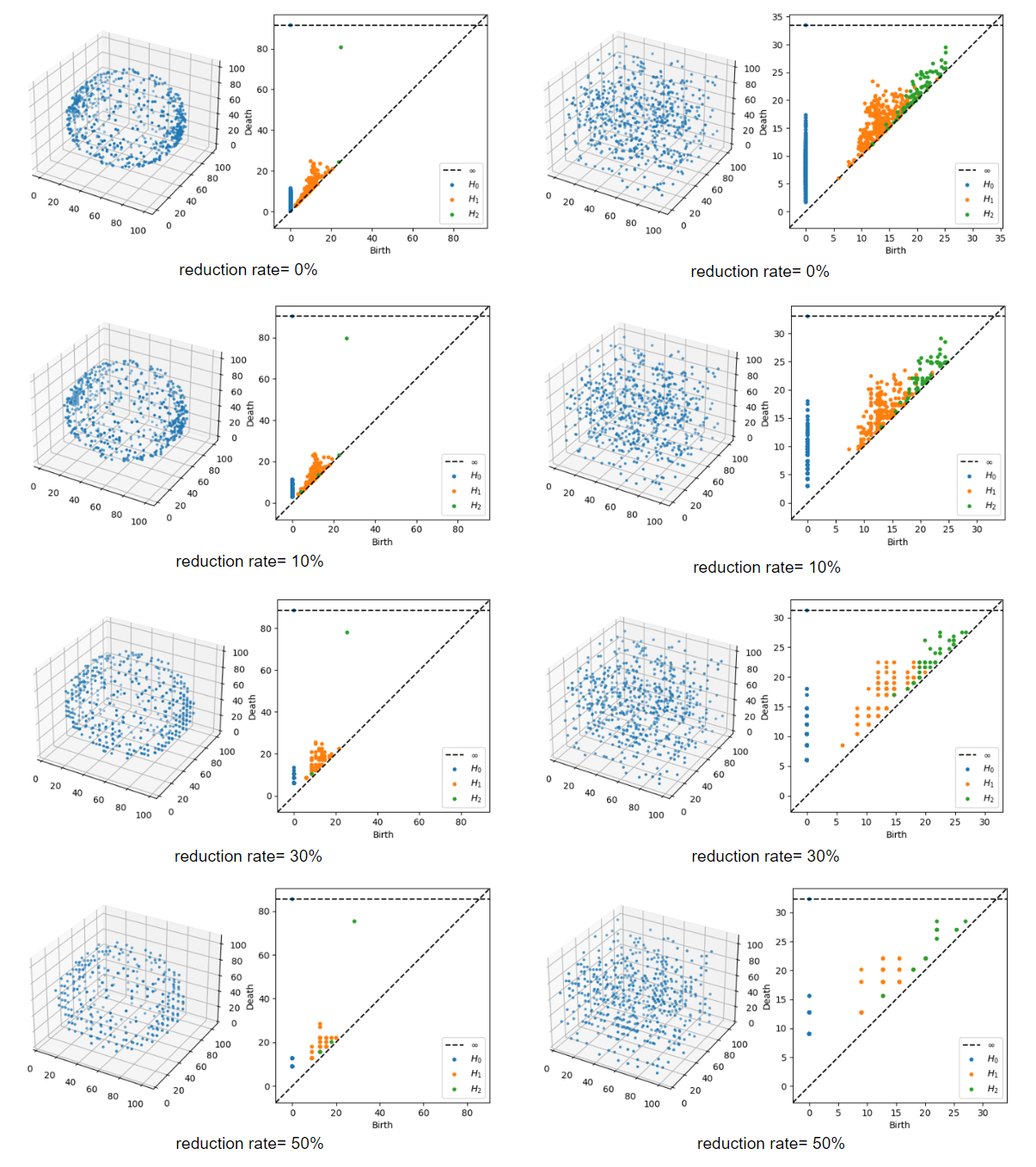}
  \caption{{\bf Evolution of the PDs for a 3D dataset} The PDs of point clouds of a 3D dataset applying CLA with the reduction rates $10\%, 30\%$, and $50\%$.}  \label{Data_PD_CLA_3D}
\end{figure}


Second, we investigate the preservation of classification accuracy when applying CLA.
For this purpose, we compare the classification accuracy derived from the following three different data processing methods using the dataset with $n = 10,000$ presented in Section \ref{Data}.
One is to classify the sphere type and random type point clouds using SVM alone, another is to classify their PDs using SVM, and the other is to classify the PDs of the reconstructed point clouds obtained by applying CLA using SVM.
The classification accuracy obtained from the 2D and 3D datasets is depicted in two graphs, denoted as (A) and (B), in Figure \ref{Fig_acc}, respectively. In both datasets, the number of points $p$ increases by $100$ from $100$ to $1,000$.
\begin{figure}[h!]
  \centering
  \includegraphics[width = 10cm]{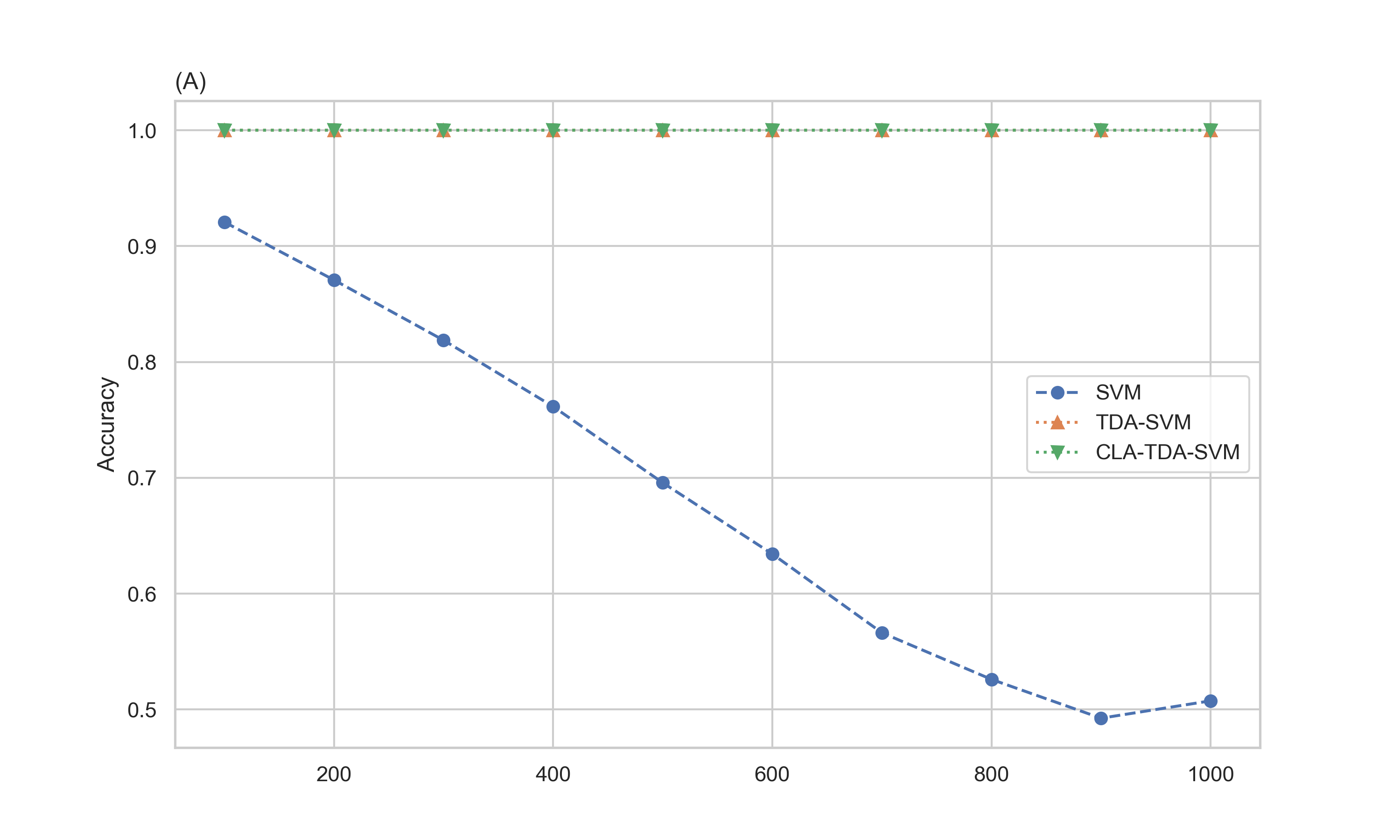}
  \includegraphics[width = 10cm]{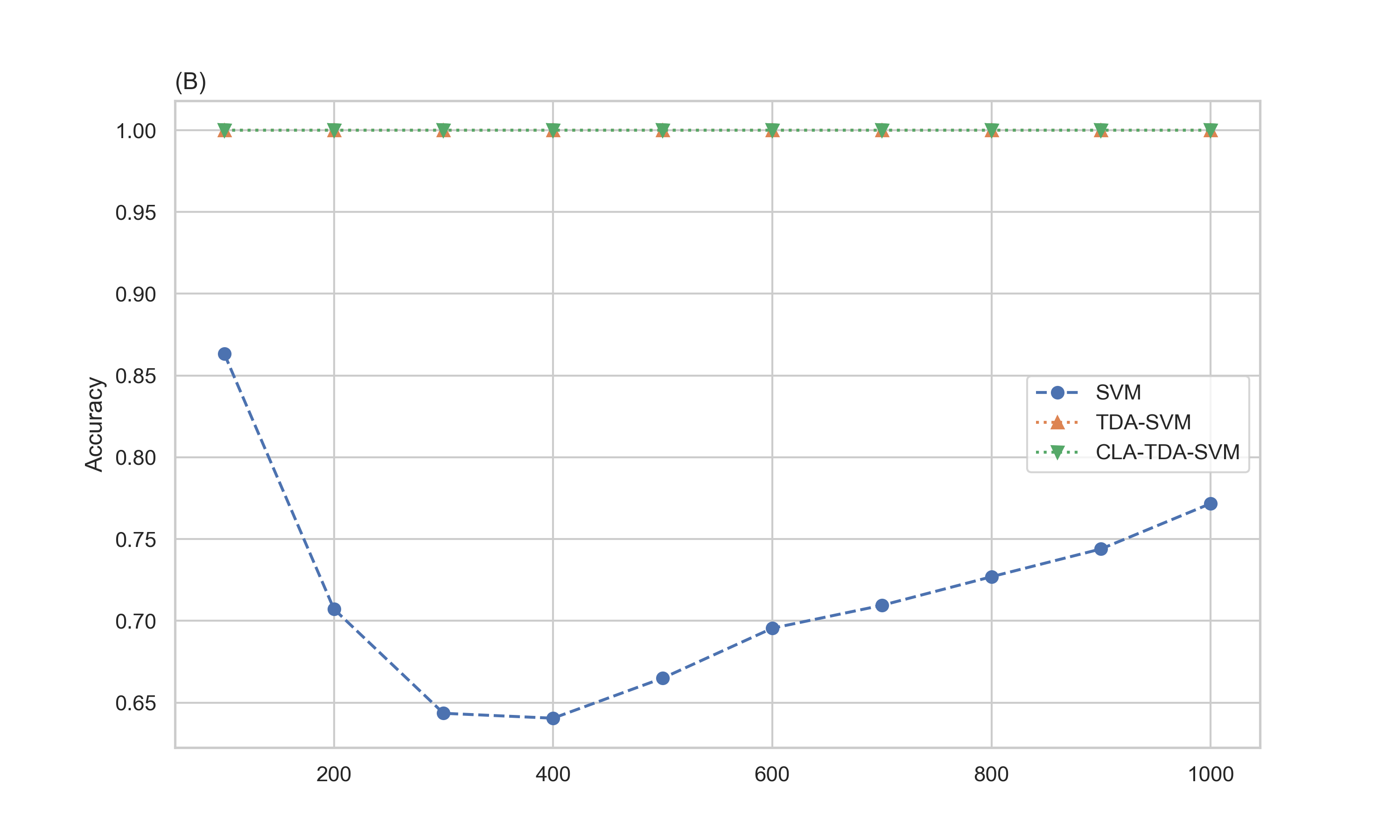}
 \caption{{\bf Classification accuracy of the 2D and 3D datasets} The above graphs show the accuracy of classification for sphere type point clouds and random type point clouds in the 2D and the 3D datasets, respectively, with $n=10,000$, where $p$ increases from $100$ to $1,000$ by $100$.
 The blue dashed lines represent the classification accuracy of the datasets using SVM alone.
 The orange dotted lines show the classification accuracy obtained by classifying the PDs of the datasets using SVM.
 Finally, the green dotted lines represent the classification accuracy determined by classifying the PDs of the reconstructed datasets obtained by applying CLA with the reduction rate $50\%$ using SVM.}
 \label{Fig_acc}
 \end{figure}
The blue dashed lines indicate the classification accuracy when only SVM is applied to the given datasets. The orange dotted lines present the classification accuracy of applying SVM to PDs obtained from the datasets. The green dotted lines show the classification accuracy of applying SVM to PDs of reconstructed datasets obtained by applying CLA with the reduction rate $50\%$ to the dataset.
When using only SVM for classification, it can be observed from graph (A) and graph (B) in Figure 10 that the classification accuracy of the 2D dataset is approximately $0.5$, and the classification accuracy of the 3D dataset is approximately $0.77$ when the number of points in the dataset is $p=1,000$.
On the other hand, when applying SVM to the PDs obtained from the datasets (i.e., utilizing TDA), the classification accuracy consistently reaches a value of $1$ in both the 2D and 3D datasets. Furthermore, it can be observed that even after reducing the number of points in both the 2D and 3D datasets by applying CLA, the classification accuracy still remains at $1$. This observation highlights the fact that CLA is a data reduction technique that effectively preserves the topological features of the data.



\vspace{1cm}
\section{\bf{Conclusion}}\label{Conclusion}

The research results can be summarized as follows:
It is natural that as the reduction rate increases, the corresponding time decreases accordingly. 
Table \ref{Table2Dtime} indicates that the use of CLA reduces the PD calculation time by approximately $93\%$ compared to calculations without CLA.
The disparity in computational cost becomes more pronounced when analyzing the 3D data.
Especially at $p=3,000$, the assessment of the computational cost without CLA was impeded by our computer specifications.
When CLA is applied with $r=50\%$, the rates of decrease for the 3D data show that the computational cost is up to $95\%$ less than calculations without CLA, as presented in Table \ref{Table3Dtime}. Notably, the reduction in computational cost is higher for the 3D data compared to the 2D data.

The application of TDA in combination with SVM for the sphere type and random type point clouds resulted in significantly higher classification accuracy compared to using SVM alone, as observed in Figure \ref{Fig_acc}.
Although CLA may slightly deform the given data as a preprocessing technique, there was no significant difference in classification accuracy when comparing the combined use of TDA and CLA with using TDA alone.
This suggests that CLA is an effective data reduction technique for TDA, and it is also theoretically supported by Theorem \ref{Thm}.\\

TDA proposes a novel approach to big data analysis. However, persistent homology, one of the representative methodologies of TDA, is currently not feasible due to its exponentially increasing computational complexity. Data reduction is the most fundamental technique used in data science to address computational complexity problems. Thus, various data reduction methods suitable for data analysis using persistent homology have been studied. The method proposed in this paper, CLA, improves upon previous related studies by allowing for data reduction as desired while preserving the topological features of high-dimensional and complex datasets. Therefore, CLA demonstrates that it can be used as an effective tool for the extension of topology in machine learning and the future of data analysis.

\vspace{1cm}
\section*{\bf{Acknowledgement}}
All authors equally contribute this paper. 
The work of Seung Yeop Yang and Seonmi Choi was supported by the National Research Foundation of Korea(NRF) grant funded by the Korean government(MSIT)(No. 2022R1A5A1033624). 
The work of Seonmi Choi was supported by the Basic Science Research Program through the National Research Foundation of Korea(NRF) funded by the Ministry of Education (No. 2021R1I1A1A01049100). 
The work of Jeong Rye Park was supported by the Basic Science Research Program through the National Research Foundation of Korea(NRF) funded by the Ministry of Education (No. 2021R1I1A1A01057767).

\end{document}